\documentclass[letterpaper,USenglish,
thm-restate,
numberwithinsect]{lipics-v2021}
\pdfoutput=1
\nolinenumbers

\titlerunning{Linked Fates}

\authorrunning{B.~Carleton et al.}

\hideLIPIcs
\Copyright{Jane Open Access and Joan R. Public} %



\usepackage{algorithm,algpseudocode}
\algrenewcommand\algorithmicrequire{\textbf{Precondition:}}

\usepackage{bm}

\def\desclabel#1{\bf #1\hfil}
\def\desc{\list{}{%
\setlength{\leftmargin}{0pt}
\labelwidth= \leftmargin
\advance \labelwidth by -\labelsep
\let \makelabel=\desclabel}}

\def\descHACKlabel#1{\bf #1\hfil}
\def\descHACK{\list{}{%
\setlength{\leftmargin}{0pt}
\labelwidth= \leftmargin
\advance \labelwidth by -\labelsep
\let \makelabel=\descHACKlabel}}

\newcounter{extremeleftlistcounter}
  {\begin{list}{\arabic{extremeleftlistcounter}~~~}{\usecounter{extremeleftlistcounter}%
        \setlength{\labelsep}{0pt}\setlength{\leftmargin}{0pt}%
        \setlength{\labelwidth}{0pt}\setlength{\listparindent}{0pt}}}%
  {\end{list}}

\newcounter{leftlistcounter}
  {\begin{list}{\arabic{leftlistcounter}~~~}{\usecounter{leftlistcounter}%
        \setlength{\labelsep}{0pt}\setlength{\leftmargin}{15pt}%
        \setlength{\labelwidth}{15pt}\setlength{\listparindent}{0pt}}}%
  {\end{list}}

\newlength{\filength}
\newsavebox{\gcbox}
\sbox{\gcbox}{\framebox[\filength]{\rule{0ex}{2ex}}}

\newtheorem{fact}[theorem]{Fact}
\newtheorem{property}{Property}

\newtheorem{positiveresult}{Positive Result}
\newtheorem{appendixnote}{Appendix Note}

\showhyphens{Faliszewski}

\newcount\hour  \newcount\minutes  \hour=\time  \divide\hour by 60
\minutes=\hour  \multiply\minutes by -60  \advance\minutes by \time
\def\mmmddyyyy{\ifcase\month\or Jan.\or Feb.\or Mar.\or Apr.\or May.\or June\or Jul.\or
  Aug.\or Sep.\or Oct.\or Nov.\or Dec.\fi\space\number\day, \number\year}
\def\hhmm{\ifnum\hour<10 0\fi\number\hour :%
  \ifnum\minutes<10 0\fi\number\minutes}

\makeatletter
\def\@citex[#1]#2{\if@filesw\immediate\write\@auxout{\string\citation{#2}}\fi
  \def\@citea{}\@cite{\@for\@citeb:=#2\do
    {\@citea\def\@citea{,\linebreak[0]}\@ifundefined
       {b@\@citeb}{{\bf ?}\@warning
       {Citation `\@citeb' on page \thepage \space undefined}}%
\hbox{\csname b@\@citeb\endcsname}}}{#1}}
\newcommand{\singlespacing}{\let\CS=
\@currsize\renewcommand{\baselinestretch}{1}\tiny\CS}
\newcommand{\singlespacingplus}{\let\CS=
\@currsize\renewcommand{\baselinestretch}{1.25}\tiny\CS}
\newcommand{\doublespacing}{\let\CS=
\@currsize\renewcommand{\baselinestretch}{1.75}\tiny\CS}
\newcommand{\extradoublespacing}{\let\CS=
\@currsize\renewcommand{\baselinestretch}{1.9}\tiny\CS}
\newcommand{\nicenicespacing}{\let\CS=
\@currsize\renewcommand{\baselinestretch}{1.9}\tiny\CS}
\newcommand{\draftspacing}{\let\CS=
\@currsize\renewcommand{\baselinestretch}{2.0}\tiny\CS}
\newcommand{\hugedraftspacing}{\let\CS=
\@currsize\renewcommand{\baselinestretch}{2.4}\tiny\CS}
\newcommand{\niceonespacing}{\let\CS=\@currsize\renewcommand{\baselinestretch}{1.1}\tiny\CS}
\newcommand{\nicetwospacing}{\let\CS=\@currsize\renewcommand{\baselinestretch}{1.2}\tiny\CS}
\newcommand{\nicethreespacing}{\let\CS=\@currsize\renewcommand{\baselinestretch}{1.3}\tiny\CS}
\newcommand{\singlespacingplusplus}{\let\CS=\@currsize\renewcommand{\baselinestretch}{1.35}\tiny\CS}
\newcommand{\nicefourspacing}{\let\CS=\@currsize\renewcommand{\baselinestretch}{1.4}\tiny\CS}
\newcommand{\nicefivespacing}{\let\CS=\@currsize\renewcommand{\baselinestretch}{1.5}\tiny\CS}
\newcommand{\nicesixspacing}{\let\CS=\@currsize\renewcommand{\baselinestretch}{1.6}\tiny\CS}
\newcommand{\nicesevenspacing}{\let\CS=\@currsize\renewcommand{\baselinestretch}{1.7}\tiny\CS}
\newcommand{\niceeightspacing}{\let\CS=\@currsize\renewcommand{\baselinestretch}{1.8}\tiny\CS}
\newcommand{\niceninespacing}{\let\CS=\@currsize\renewcommand{\baselinestretch}{1.9}\tiny\CS}
\makeatother%

\newcommand{\naturalnumber}{\ensuremath{{  \mathbb{N} }}}
\newcommand{\naturalnumberpositive}{\ensuremath{{  \mathbb{N}^+ }}}



\newcommand{\acc}{\ensuremath{{\text{\rm \#acc}}}}

\newcommand{\sat}{{\rm SAT}}

\newcommand{\np}{\ensuremath{\mathrm{NP}}}
\newcommand{\ip}{{\rm IP}}
\newcommand{\up}{\ensuremath{\mathrm{UP}}}
\newcommand{\upleq}[1]{\ensuremath{\up_{{\leq}{#1}}}}
\newcommand{\upbigohone}{\upleq{\bigo(1)}}
\newcommand{\upbigohlogn}{\upleq{\bigo(\log(n))}}

\newcommand{\fewp}{\ensuremath{\mathrm{FewP}}}
\newcommand{\coup}{{\rm coUP}}

\newcommand{\p}{\ensuremath{\mathrm{P}}}

\newcommand{\conp}{{\rm coNP}}
\newcommand{\pspace}{\ensuremath{\mathrm{PSPACE}}}

\newcommand{\ph}{\ensuremath{{\rm PH}}}

\newcommand{\ceiling}[1]{{{\lceil {#1} \rceil}}}
\newcommand{\floor}[1]{{{\lfloor {#1} \rfloor}}}

\newcommand{\manyone}{\ensuremath{\,\leq_{m}^{p}\,}}

\newcommand{\sigmastar}{\ensuremath{\Sigma^\ast}}

\newcommand{\pisnp}{\ensuremath{\p=\np}}

\newcommand{\calf}{\ensuremath{{\mathcal{F}}}}
\newcommand{\calC}{\ensuremath{{\mathcal{C}}}}

\newcommand{\bigo}{{\protect\mathcal{O}}}
\newcommand{\bigoh}{{\protect\mathcal{O}}}
\newcommand{\condition}{\,\mid \:}
 
\def\land{{\; \wedge \;}}

\newcommand{\ntonplus}[1]{#1 : \naturalnumber \rightarrow \naturalnumberpositive}

\newcommand{\reals}{\mathbb{R}}
\newcommand{\positivereals}{\reals^{+}}
\newcommand{\oneupreals}{\reals^{\geq 1}}
\newcommand{\ntor}[1]{#1 : \naturalnumber \rightarrow \oneupreals}

\newcommand{\YearsSinceWatanabe}{\number\numexpr\year-1988}

\title{%
  Linked Fates: How Small of 
an Ambiguity Increase Can Make the Difference
  Between Equaling and Separating from P?}
\author{Benjamin Carleton\footnote{Work done in part while at the University of Rochester's Department of Computer Science.}}
{Department of Computer Science, Cornell University, Ithaca, NY 14850, USA}
{}{https://orcid.org/0000-0003-2369-4423}{}

\author{Michael C. Chavrimootoo\footnotemark[1]}{Department of Computer Science, Denison University, Granville, OH 43023, USA}{}{https://orcid.org/0000-0002-0510-6002}{}

\author{Lane A. Hemaspaandra}{Department of Computer Science, University of Rochester, Rochester, NY 14627, USA}{}{https://orcid.org/0000-0003-0659-5204}{}

\author{David E. Narv\'{a}ez\footnote{Work done in part while at the University of Rochester's Department of Computer Science and Virginia Tech's Bradley Department of Electrical and Computer Engineering.}}{Faculty of Mathematics and Physics, University of Ljubljana, Ljubljana, Slovenia}{}{https://orcid.org/0000-0003-3704-1060}{}

\author{Conor Taliancich\footnotemark[1]}{Property Matrix, Culver City, CA 90230, USA}{}{https://orcid.org/0000-0002-7419-0391}{}

\author{Melissa Welsh\footnotemark[1]}{Appian Corporation, 7950 Jones Branch Dr., McLean, VA 22102, USA}{}{https://orcid.org/0000-0003-4673-3376}{}

\funding{Work supported in part
  by NSF grants CCF-2006496 and 
  DUE-2135431, CIFellows grant CIF2020-UR-36, a Renewed Research Stay grant from the Alexander von~Humboldt Foundation, and the William G. and Mary Ellen Bowen Endowed Fund.}

\acknowledgements{%
We thank Lance Fortnow for his oracle pointer that is discussed in Section~\ref{s:related}.}

\begin{document}

\sloppy

\maketitle
\begin{abstract}
  Ambiguity-bounded versions of NP, denoted $\upleq{f(n)}$, bound by
  $f(n)$ the number of accepting paths the nondeterministic
  polynomial-time Turing machine can have on inputs of length $n$.
  Such classes range from Valiant's completely unambiguous ($f(n)=1$)
  class UP to NP itself, where there is no bound or, equivalently,
  there is the toothless exponential bound ($f(n) = 2^{n^{\bigo(1)}}$).

  This paper seeks to understand which of these classes stand and fall
  together as to whether they equal deterministic polynomial time.
  Informally put, what ranges of ambiguities have linked fates?  That
  is, for which pairs of 
  nondecreasing 
  functions, $(f_1 ,f_2)$,
  satisfying $(\forall n)[f_1(n) \leq f_2(n)]$, does it hold that
  $\p = \upleq{f_1(n)} \implies \p = \upleq{f_2(n)}$. More
  particularly, for which pairs does that hold robustly, i.e., 
  it holds in the
  real world and every relativized world?  And for which pairs
  does that implication fail to hold robustly, i.e., there is an oracle
  $A$ such that $\p^A = \upleq{f_1(n)}^A \subsetneq \upleq{f_2(n)}^A$?

  The only previously known positive result is Watanabe's  
  1988 
  result
  that $ \p = \upleq{1} \implies \allowbreak (\forall k \geq 1)[\p = \upleq{k}]$, which 
  even 
  holds robustly. His result, though lovely, applies only to constant-bounded ambiguities. As our positive result, we present a new class of cases (Theorem~\ref{t:positive-strengthening}) that apply (and even robustly apply) at greater ambiguity levels. 
  To give our class of cases, we leverage two approaches: 
  a novel path-poisoning approach that works even on superconstant ambiguities (Theorem~\ref{t:constant-addition})
  and 
  a new application
  of 
  the power of 
  padding 
  (Theorems~\ref{t:padding}/\ref{t:poly-stretch}).
  As negative results, we show
  that for essentially all other cases, no linkage holds robustly.

\end{abstract}

\section{Introduction}

The bounded-ambiguity class $\upleq{f(n)} \subseteq \np$ is the set of languages accepted by NPTMs having at most $f(n)$ accepting paths on each input (Section~\ref{s:defs} gives a more 
detailed definition).
For example, using ``1'' as a shorthand for the constant function $f(n)=1$,
one gets 
the 
class $\upleq{1}$, which was introduced in seminal work of Valiant~\cite{val:j:checking} (where it 
is called $\up$, which 
it still normally goes by)
and is often described as  capturing ``unambiguous nondeterministic computation.''
$\up$ has many surprising connections. 
Perhaps most notably, $\up$ has been central in characterizing whether 
worst-case one-way functions and one-way permutations exist.
We briefly discuss that relation in Section~\ref{s:defs} and refer interested readers to~\cite{%
gro-sel:j:complexity-measures,ko:j:operators,%
hem-rot:j:aowfs,hom-tha:j:owp}, or for a more self-contained overview of the relationship between ambiguity-bounded versions of $\np$ and one-way functions, see~\cite[Chapter~2]{hem-ogi:b:companion}. 
From a complexity perspective, we are primarily concerned with the relationship established by Watanabe (\cite{wat:j:hardness-one-way}; in that it is phrased in the language of one-way functions, rather than in the equivalent complexity-class form used here, see~\cite[Chapter~2]{hem-ogi:b:companion}) that for each $k \in \naturalnumberpositive$, $\p=\upleq{1} \iff \p = \upleq{k}$.

In particular, this paper is concerned with which pairs of functions yield ambiguity-bounded versions of $\np$ with ``linked fates,'' i.e., for which nondecreasing functions $f_1$ and $f_2$ such that $(\forall n \in \naturalnumber)[f_1(n) \leq f_2(n)]$ does it hold that $\p = \upleq{f_1(n)} \iff \p=\upleq{f_2(n)}$.

Watanabe's result~\cite{wat:j:hardness-one-way} settles the case where $f_1$ and $f_2$ are constant bounded, i.e., there is a $k \in \naturalnumberpositive$ such that for all $n \geq 0$, $f_1(n) \leq f_2(n) \leq k$.
However, over the last~\YearsSinceWatanabe\ years, there has been no success in strengthening the work of Watanabe beyond the constant case. 
For example, for what type of function $\ntor{f}$ does it hold that $\p=\up \implies \p=\upleq{f(n)}$? Beyond the obvious case where $f$ is constant-bounded, we first prove in this paper that for no 
monotonically 
nondecreasing and unbounded function $\ntor{f}$ can $\p=\up \implies \p=\upleq{f(n)}$ hold robustly, that is, for a given such function $f$  we provide an oracle $A$ such that $\p^A=\up^A \neq \upleq{f(n)}^A$.

As our positive result, we define a new types of linked fates, which we discuss below at an informal level. 
We will need the following notation.
For any class $\calf$ of
functions each mapping from $\naturalnumber$ to
$\oneupreals$,
$\upleq{\calf} = \{ L \condition (\exists f \in \calf)[L \in
\upleq{f(n)}]\}$. (Since we view big-Oh notations as giving a 
class of functions, and with that extending even through ``$+$''
operators/etc., having classes such as $\upleq{\bigoh(1)}$ or, as below, 
 $\upleq{f(n) + \bigoh(1)}$ is meaningful; in this notation, having a big-Oh 
 within a ``$\leq$'' is not a redundancy.)

\begin{positiveresult}[Theorem~\ref{t:positive-strengthening}]
For each monotonically nondecreasing function $\ntor{f}$ and each $k\in \naturalnumberpositive$, $\p=\upleq{f(n)} \iff \p=\upleq{f(n^k) + \bigoh(1)}$.
\end{positiveresult}

Our result highlights two things. 
First, this result implies Watanabe's result~\cite{wat:j:hardness-one-way}, but is not a direct generalization of Watanabe's approach. Indeed, the obvious way to generalize the technique used by Watanabe seems limited to functions $f$ where $f(n)$ is in $n^{\bigoh(1)}$, but our approach successfully bypasses that limitation.
Second, this result highlights that a polynomial increase in the argument to $f$ is not enough to separate from $\p$\@. 

The result itself relies on a combination of novel ``path-poisoning'' approach (Theorem~\ref{t:constant-addition}) and 
a simple, yet powerful, padding approach (Theorems~\ref{t:padding}/\ref{t:poly-stretch}).
Naturally, one might consider why the assumption on $f$ is important. We give the following example to demonstrate the importance of the assumption.
\begin{example}
Let $B$ by a fixed NP-complete set that has no strings at lengths 
that are perfect squares (i.e., $(\forall x \in B)[|x| \not\in \{0,1,4,9,16,\dots\}]$). It is easy to see  that 
such sets exist (see also Appendix Note~\ref{app:encoding}).

    Now consider the function $f(n) = \begin{cases}
        1 & \text{if $n$ is a perfect square} \\
        2^n & \text{otherwise}.
    \end{cases}$

Then $B \in
 \upleq{f(n)}$ and $\upleq{f(n^2) + \bigoh(1)} = \upleq{\bigoh(1)}$. 
    So (because (a)~since $B$ is NP-complete, we have that $B\in \upleq{f(n)} \implies (\pisnp \iff \p = \upleq{f(n)})$, and (b)~by Watanabe's result mentioned above,
    $\p=\upleq{\bigoh(1)}   \iff \p = \up$), 
    giving a proof that $\upleq{f(n)}$ and $\upleq{f(n^2)+\bigoh(1)}$ have linked fates is equivalent to giving a proof that  ``$\p = \up \iff \p=\np$,'' which is a longstanding open problem (see also our Conclusion section for more details). Further, since there is an oracle relative 
    to which $\p=\up\neq \np$~\cite{rac:j:rel}, it is impossible to 
    robustly prove that $\upleq{f(n)}$ and $\upleq{f(n^2)+\bigoh(1)}$ have linked fates.
\end{example}

We also give many applications of Theorem~\ref{t:positive-strengthening} via our Corollary~\ref{c:concrete-examples}.
Finally, we provide a flexible result (via Theorem~\ref{t:stretch-oracle}) that in one fell swoop shows that our linked fates, in addition to those of Watanabe~\cite{wat:j:hardness-one-way}, are the only ones that hold robustly, i.e., in the real world and in every relativized world.

\section{Definitions}\label{s:defs}
Let $\naturalnumber = \{0, 1, 2, \ldots\}$, let $\naturalnumberpositive = \{1, 2, 3, \ldots\}$, let $\reals$ denote the set of real numbers, let $\oneupreals = \{x \in \reals \mid x \geq 1\}$, and let $\positivereals = \{x \in \reals \mid x > 0\}$.
As is standard, given two functions $f$ and $g$ from $\naturalnumber$ to $\positivereals$, we say that $f(n) = \bigoh(g(n))$ exactly if there are positive natural numbers 
$c$ and $n_0$
such that $(\forall n \geq n_0)[f(n) \leq cg(n)]$~\cite{sip:b:introduction-third-edition}.
A function $f : \naturalnumber \rightarrow \positivereals$ is monotonically nondecreasing if $(\forall n_1, n_2 \in \naturalnumber)[n_1 \leq n_2 \implies f(n_1)\leq f(n_2)]$. Moreover, whenever we write $\log(\cdot)$, we mean $\log_2(\cdot)$.

For any nondeterministic polynomial-time Turing machine (NPTM) $N$
and any string $x$, 
$\acc_N(x)$ denotes the number of accepting computation paths
of $N$ on input $x$.
We will also use this notion when an oracle $A$ is present, e.g., $\acc_{N^A}(x)$.
As is standard,
$B \manyone C$ ($B$ polynomial-time many-one
reduces to $C$) means that there is a
polynomial-time computable function $h$ such that for
each $x$ it holds that $x\in B \iff h(x)\in C$. 

Valiant~\cite{val:j:checking} started the study of ambiguity-bounded
Turing machines by defining the class $\up$.  A language $L$ belongs
to $\up$ exactly if there is an
NPTM, $N$, accepting $L$ such that for each
$x$ it holds that $\acc_N(x)\leq 1$.  So when the machine accepts, it does so
on exactly one computation path. 

Since Valiant's seminal work, many other ambiguity-bounded classes have
been defined and studied, starting with the work of 
Allender and Rubinstein~\cite{all-rub:j:print}.
Varying notations have been used in the literature. 
The notation that we will use for such classes is
the very natural one due to 
Lange and
Rossmanith~\cite{lan-ros:c:up-circuit-and-hierarchy}:
If $f$ maps from $\naturalnumber$ to $\oneupreals$, then
$\upleq{f(n)}$ will denote all languages $L$ such that there exists an NPTM,
$N$, such that for each $x\notin L$ it holds that $\acc_N(x) = 0$,
and for each $x\in L$ it holds that $1 \leq \acc_N(x) \leq f(|x|)$.
Valiant's class $\up$ is, in this notation, $\upleq{1}$.

We will appeal to the following fact, which follows directly from the above definition of $\upleq{f(n)}$, to establish our positive results.

\begin{fact}\label{f:trivial}
  If $f_1$ and $f_2$ are functions from $\naturalnumber$ to $\oneupreals$ such that
  for all but finitely many values of $n\in \naturalnumber$ it holds that $f_1(n) \leq f_2(n)$, then $\p=\upleq{f_2(n)} \implies \p=\upleq{f_1(n )}$.
\end{fact}

For any class $\calf$ of
functions each mapping from $\naturalnumber$ to
$\oneupreals$,
$\upleq{\calf} = \{ L \condition (\exists f \in \calf)[L \in
\upleq{f(n)}]\}$.  
Among the existing ambiguity-bounded classes are 
the polynomial-ambiguity class 
$\fewp = \bigcup_{k \in \naturalnumberpositive} \upleq{n^k + k}$
of Allender
and Rubinstein~\cite{all-rub:j:print}; the log-bounded ambiguity class
$\upleq{\bigoh(\log(n))}$ and also the classes $\upleq{\bigoh(\sqrt{\log(n)})}$ and
$\upleq{\bigoh(\log(\log(n)))}$
of Hemaspaandra et al.~\cite{hem-juv-nad-phi:j:ics};
and the $k$-bounded ambiguity classes $\upleq{2}$, $\upleq{3}$,
\dots~of Beigel~\cite{bei:c:up1} (that paper used the 
slightly different notation $\up_2$, $\up_3$,~\dots\ for these classes), and the 
the
constant-bounded ambiguity class 
$\upleq{\bigoh(1)}$ of 
Hemaspaandra and
Zimand~\cite{hem-zim:tOutByJourExceptUPkStuffOnlyIsHere:balanced}
(who defined it as the union $\upleq{1} \cup 
\upleq{2} \cup \upleq{3} \cup \dots$, but that is easily 
seen to yield the same class as 
$\upleq{\bigoh(1)}$).
As
mentioned in the abstract, even NP is a toothless case of this notation,
namely 
$\np = \upleq{2^{n^{\bigo(1)}}}$.\footnote{%
There is a rather subtle issue at play in this paragraph. 
Suppose we were to speak of  the 
class $\upleq{\log(n)}$. Its definition blows up for $n=0$
and for $n=1$ it would exclude any set that had any strings 
of length one. The former is terrible, and the latter 
is unnatural and problematic since it excludes a countably infinite number of sets in $\p$ from $\upleq{\log(n)}$, and so even the statement $\p=\upleq{\log(n)}$ would be false. To remedy this situation, we implicitly assume in this work that ``$\log(\cdot)$''
when used in our ambiguity bounds is a shorthand for ``$\log(\max(2, \cdot))$.''}

Each of these classes is also motivated by the fact that the equality
of each such ambiguity-bounded class with P is characterizing whether
complexity-theoretic one-way functions exist whose ambiguity (i.e.,
whose limits on the number of preimages of each string in the range)
shares the same bound.  Without going into detail here about the
definitions of one-way functions and their ambiguity, we simply
mention 
the general correspondence that holds:
$\p \neq \upleq{f(n)}$ if and only if there exists an $f(n)$-to-1
one-way function 
(\cite{hem-juv-nad-phi:j:ics} states that general
case, but the specific cases started as early as Grollmann and
Selman's~\cite{gro-sel:j:complexity-measures} 
and Ko's~\cite{ko:j:operators} result that $\p\neq\upleq{1}$
if and only if one-to-one one-way functions exist;
see
also~\cite{all-rub:j:print,%
hem-zim:tOutByJourExceptUPkStuffOnlyIsHere:balanced} and 
\cite[Chapter 2]{hem-ogi:b:companion}).

For any claim regarding classes for which a generally agreed notion of
relativization
has been reached in the literature, we say
that the claim
holds \emph{robustly} if it holds for each oracle~$A$.  For example, if
we were to say that $\p \subseteq \np$ holds robustly---and indeed it
does---that would be asserting that for each oracle $A$ it holds that
$\p^A\subseteq \np^A$.  The vast majority of theorems, proofs, and proof
techniques used to study complexity classes relativize.  Thus, showing
that a claim holds in some relativized world establishes that a wide
range of proof techniques cannot prove that the claim fails in the
``real'' (i.e., unrelativized) world.  However, there exist 
proof techniques and results, e.g., arithmetization and
$\ip = \pspace$, that (in some cases controversially as to what the
``right''/``fair'' model of relativization is) seem not to
relativize~\cite{bab-for:j:arithmetization,cha-cha-har-ran-roh:j:revisionist,%
for-kar-lun-nis:j:ip,sha:j:ip,buh-for-thi:c:nonrelativizing,%
ver:j:oracle-survey,aar:url:MIPstar-nonrelativizing-blog-post-4512,%
ji-nat-vid-wri-yue:t:mip*=re}.

Throughout this paper, $M_1$, $M_2$, $M_3$,~$\ldots$ denotes some standard enumeration of polynomial-time deterministic (oracle) Turing machines, such that for all oracles $A$ it holds that the deterministic running time of $M_i^A$ 
on an input of size $n$ 
is bounded by $n^i+i$, and $N_1$, $N_2$, $N_3$,~$\ldots$ denotes some standard enumeration of polynomial-time nondeterministic (oracle) Turing machines, such that for all oracles $A$ it holds that the nondeterministic running time of $N_i^A$ 
on an input of length $n$
is bounded by $n^i+i$. We here are describing this standard setup (see, e.g., \cite{har-imm:cOUTDATEDseegur:c:compOUTDATED:complete,har-hem:j:up}) using almost verbatim the wording in Hemaspaandra, Jain, and Vereshchagin~\cite{hem-jai-ver:j:up-turing}. However, the idea underpinning why such enumerations exist---the idea of in effect pairing each machine with a set of polynomial running times such that every polynomial is majorized by at least one of them---is inspired by work going at least as far back as Sipser's paper on the nonexistence of complete sets for certain classes~\cite{sip:c:complete-sets}.

The function  $\ntonplus{e}$
is defined inductively as follows: $e(0) = 2$ and for each 
$n \in \naturalnumber$, $e(n+1) = 2^{2^{e(n)}}$. 
One could visualize 
$e$
as mapping from 
$\naturalnumber$
to certain towers of 2s, and so
when we speak of a ``tower length'' we 
mean a length $m$ that is mapped onto by~$e$.

A Turing machine that on each input has at most one accepting path is said to be a categorical Turing machine. Similarly,
for each $\ntor{f}$, for each oracle $A$, and for each 
nondeterministic oracle Turing machine $N$,
we will say
that $N^A$
is $f$-categorical if $N^A$ has at most $f(n)$ accepting paths on each input of length~$n$. Otherwise, $N^A$ is not $f$-categorical.
In an
abuse of language and terminology,
we also say that an oracle machine (resp.\ pairs of oracle machines) robustly has (resp.\ have) a property if that property holds 
relative to every oracle.%
\footnote{This notion subtly differs from our original notion of robustness, 
which
holds for statements about complexity classes, and here we are applying the notion to properties of oracle machines.}
So for example, ``nondeterministic oracle Turing machine $N_i$ is robustly $f$-categorical'' means that for every oracle $A$, $N_i^A$ is $f$-categorical.

\section{Linked and 
Unlinked Fates}\label{s:results}
In this section, we will discuss two differing approaches
to showing that pairs of classes have linked fates.
For each, we give relativized results showing that if one goes even
slightly beyond what that approach shows to be linked (and indeed
robustly linked), one can build oracles in which the classes' fates
fail to be linked.  That is, we show that our
positive results are optimal with respect to relativizable
proof techniques, i.e., with respect to robust results.

Our first type of linkage is one implicitly achieved by Watanabe, who showed that unambiguous one-way functions exist if and only
if constant-ambiguity one-way functions
exist~\cite{wat:j:hardness-one-way}.  Viewed through the lens of the
relationships between one-way functions and ambiguity-bounded
classes~\cite{gro-sel:j:complexity-measures,%
  hem-zim:tOutByJourExceptUPkStuffOnlyIsHere:balanced,hem-ogi:b:companion},
that equivalently establishes the following.
\begin{theorem}[\cite{wat:j:hardness-one-way}]\label{t:osamu}
For each $k \geq 1$, it holds that
$\p=\upleq{1} \iff \p=\upleq{k}$.
\end{theorem}
The treatment
in~\cite{hem-ogi:b:companion} proves that implication directly within
the language of ambiguity-bounded classes, rather than indirectly
through one-way functions.
Theorem~\ref{t:osamu} shows that for every two
constants $k_1 \in \naturalnumberpositive$ and
$k_2 \in \naturalnumberpositive$, $k_1 \leq k_2$, it holds that the
fates of $\upleq{k_1}$ and $\upleq{k_2}$ are linked.  Or, taking the
extreme case, it says that 
the fates of $\upleq{1}$ and $\upbigohone$ are linked:
$\p=\upleq{1} \iff \p=\upbigohone$.
All of those proofs relativize, so the results hold both in the real
world and in every relativized world.

We now show that that is as far as the linked fates go in this case.
We say that a function is unbounded if for each
$k\in \naturalnumberpositive$ there exists a
$j\in\naturalnumber$ such that $f(j) \geq k$.
For any function $\ntor{f}$ that is unbounded, one cannot robustly
prove that the corresponding class's fate is linked to that of $\upleq{1}$.
We require $f$ to be unbounded as otherwise, if it were bounded by a 
constant $k$,
by the fact that Theorem~\ref{t:osamu} relativizes, it would hold that for 
each oracle $A$, $\p^A = \up^A \implies \p^A=\upleq{k}^A=\upleq{f(n)}^A$,
thereby making our theorem impossible to prove (as it would be false).
We also require $f$ to be monotonically nondecreasing and in $2^{n^{\bigoh(1)}}$ so as to be able to appeal to Theorem~\ref{t:stretch-oracle} (and since $\np = 2^{n^{\bigoh(1)}}$, there are no obvious advantages to looking at cases where $f(n) \neq 2^{n^{\bigoh(1)}}$).

\begin{theorem}\label{t:oracle-super-constant}
	Let $\ntor{f}$ be a 
    monotonically nondecreasing
    function in $2^{n^{\bigoh(1)}}$ that is unbounded.  Then there exists an
	oracle $A$ such that $\p^A = \upbigohone^A \subsetneq \upleq{f(n)}^A$.
    Moreover, if the function $\floor{f(n)}$ is computable, then $A$ is recursive.
\end{theorem}
\begin{proof}
This result follows directly from 
Theorem~\ref{t:stretch-oracle} (which we give later), by considering the case where $f_1(n) = 1$ and $f_2(n) = f(n)$. These functions satisfy the requirements of Theorem~\ref{t:stretch-oracle}, 
and since the proof of Theorem~\ref{t:osamu} relativizes, it holds that there is an oracle $A$ that satisfies the conclusion of the current theorem. And if the function $\floor{f(n)}$ is computable, it follows that the oracle is recursive.
\end{proof}

The above result, Theorem~\ref{t:osamu} of Watanabe, is the only
result we know of in the literature that implies a linked-fates
situation for polynomial-time, ambiguity-bounded nondeterminism.
However, we now give a completely different family of such linked-fates
cases.  We do so via the power of padding---a technique that has been
central in complexity theory
in a  remarkably varied range of settings, such as  the elegant
construction  of universal  complete  sets  for many  complexity
classes~\cite{har:b:feasible-provable},    the    study    of    whether
all NP-complete sets  are  isomorphic  to  SAT~\cite{ber-har:j:iso},  and  
the connection between polynomial-time and exponential-time complexity
classes~\cite{boo:j:tally,har-imm-sew:j:sparse}---%
to reduce how
much nondeterminism is needed relative to the input length.

\begin{theorem}\label{t:padding}
  For each $k \in \naturalnumberpositive$ and each 
  function 
  $\ntor{f}$,
  $\upleq{f(n^k)} \manyone \upleq{f(n)}$.
\end{theorem}

\begin{proof}
  Let $L$ belong to $\upleq{f(n^k)}$.  Let $a$ be a character in
  $L$'s alphabet, $\Gamma$.

  We claim that
  $L \manyone \{x \cdot a^{|x|^k - |x|} \condition x \in L\}$, where
  $\cdot$ denotes concatenation, and $a^j$ denotes the string
  consisting of $j$ many ``$a$''s.  In particular, the polynomial-time
  many-one reduction simply is the mapping from $x$ to
  \mbox{$x \cdot a^{|x|^k - |x|}$}.
    
  However,
  $\{x \cdot a^{|x|^k - |x|} \condition x \in L\} \in \upleq{f(n)}$.
  In particular, on an arbitrary input $y$, the $\upleq{f(n)}$ machine
  will parse the input into $x$ and $a^{|x|^k - |x|}$, or if the input
  cannot be parsed as that, will reject as the input is clearly not in
  $\{x \cdot a^{|x|^k - |x|} \condition x \in L\}$.  Since $n^k-n$ is
  nondecreasing, it is impossible that a given input string can have
  two different $x$ values (say, at different lengths) that can be
  parsed as being its $x$ value.  Our $\upleq{f(n)}$ machine then
  simulates the $\upleq{f(n^k)}$ machine for $L$, running
  on input $x$.  Since
  $| x \cdot a^{|x|^k - |x|} | = |x|^k$, the $f(|x|^k)$ ambiguity
  bound that the $\upleq{f(n^k)}$ machine for $L$ will have in its
  simulated run ensures that our machine will itself satisfy an $f(n)$
  bound relative to its input size, since its input size itself
  is~$|x|^k$.
\end{proof}

\begin{restatable}{theorem}{polystretch}
\label{t:poly-stretch}
  For each $k \in \naturalnumberpositive$ and each 
  monotonically nondecreasing function $\ntor{f}$,
  it holds  that $\p = \upleq{f(n)} \iff \p = \upleq{f(n^k)}$.
\end{restatable}

Theorem~\ref{t:poly-stretch} provides the first result on ``linked fates'' with respect to ambiguity-bounded versions of $\np$ since the 1988 work of Watanabe~\cite{wat:j:hardness-one-way}.
For more common functions such as $n$ and $2^n$, this padding argument yields even more dramatic collapses (see Corollary~\ref{c:concrete-examples}), because, informally, increasing the input size effectively permits us to have more accepting paths. However, when dealing with slow-growing functions (e.g., $\log(\log(\log(n)))$), padding does not seem to buy us much. In our next theorem, we show how to get an additional constant number of accepting paths ``for free.''

In the proof of Theorem~\ref{t:constant-addition} we observe that $\upleq{f(n) + \bigoh(1)} = \bigcup_{j \in \naturalnumber} \upleq{f(n)+j}$, and so the proof of that theorem
is established by induction on $j$, which is interesting in its own right due to the nature of the argument used in the inductive step. We give an intuitive overview of that argument at the start of the inductive step, but in a nutshell, the idea is to ``poison'' an accepting path of an NPTM $M$ accepting a language in $\upleq{f(n)+j+1}$,
and then use the assumption that $\p=\upleq{f(n)+j}$ in a two-step ``attack'' that leverages the ``poisoning'' insight to drive the conclusion that $L(M) \in \p$.

\begin{theorem}\label{t:constant-addition}
  For each monotonically nondecreasing function $\ntor{f}$, $\p = \upleq{f(n)} \iff \p = \upleq{f(n)+\bigoh(1)}$.
\end{theorem}
\begin{proof}
    Let $f$ satisfy the conditions of the theorem.
    The right-to-left direction holds by Fact~\ref{f:trivial}. The rest of this proof is about the left-to-right direction. Assume that $\p = \upleq{f(n)}$.
    Given  the definition of big-Oh, it is easy to see that $\upleq{f(n) + \bigoh(1)} = \bigcup_{j \in \naturalnumber} \upleq{f(n)+j}$. So let $j \in \naturalnumber$. It suffices to show that $\p = \upleq{f(n)+j}$, which we prove by induction on $j$.

    The base case (when $j=0$) is trivial, so let us now assume that $j > 0$.
    For the inductive step, we let the inductive hypothesis be $\p = \upleq{f(n)+j}$ and assume it is true. We will show that $\p=\upleq{f(n)+j+1}$.
    Pick an arbitrary $L \in \upleq{f(n) + j +1}$, and let $M$ be an NPTM that accepts $L$ and has at most $f(n)+j+1$ accepting paths on each input of size $n$, i.e., $(\forall x)[(x \in L \implies 1 \leq \acc_M(x) \leq f(|x|) +j+1) \land (x\not\in L \implies \acc_M(x) = 0)]$). 

    Informally, the intuition for the two-step approach we employ is as follows. We wish to ``poison'' exactly one accepting path of $M$ for each input $x\in L$ with $\acc_M(x) > 1$---that is, we wish to have exactly one accepting path of $M$ on $x$ become a rejecting path whenever $M$ has more than one accepting path on input $x$. In doing so, we would be decreasing $M$'s number of accepting paths on each input with multiple accepting paths (without changing the language it accepts), thus bounding the number of accepting paths by $f(|x|) + j$ for each $x\in L$. This is however not a 
    realistic goal
    as doing so seems to require the NPTM to know when it has more than one accepting path. But we can draw insight from this scenario.  Consider now this informal description of the machine $M'$ whose behavior is the same as that of $M$, except that on each input $x\in L$, it has exactly one of its accepting path(s) ``poisoned'' (and we do not specify the ``how''). The language of $M'$ is 
     $R = \{x \in L \mid \acc_M(x) > 1\}$,
    which is a generalization of one used to prove Theorem~\ref{t:osamu} in~\cite{hem-ogi:b:companion}.  Our biggest task is to prove that $R\in\p$ under the current assumptions, which will be crucial to us in proving that $L \in \p$. 
    Unfortunately, we cannot use the same approach used by~\cite{hem-ogi:b:companion}; in that approach, an NPTM only needs to guess a constant number of distinct paths, but here, the number of paths may be too large for an NPTM to guess. So we use a different approach to prove that $R \in \p$, and then use that fact to conclude that $L\in\p$.
    
    To prove that $R \in \p$, we first prove that the language $S$ (defined below) is in $\p$.\footnote{In this proof, we tacitly assume without loss of generality that we have access to a polynomial-time computable total order over the accepting paths of Turing machines. Indeed, one can encode an accepting path as a 
    string over a finite alphabet (with at least two elements) using only polynomial amount of time (and thus space). And the standard lexicographical ordering over those strings is a total order that satisfies our assumption.}
        $S = \{(x, p) \mid x \in L$ and $p$ is an accepting path of $M$ on $x$ and $M$ has an accepting path on input $x$ that is lexicographically smaller than $p\}$.

    \begin{lemma}\label{l:s-in-p}
        $S \in \p$.
    \end{lemma}
    \begin{proof}
    Let $N_S$ be an NPTM that does the following on input $z$.
    If $z$ is not a pair, reject. Otherwise, let $(x, p)$ be the pair represented by $z$. If $p$ is an accepting path of $M$ on $x$, then guess a path $p'$ of $M$ on $x$ that is lexicographically smaller than $p$ (if no such path exists, then reject). If $p'$ is an accepting path, then we accept on the current path. If $p'$ is not an accepting path, then we reject on the current path.

    We first argue correctness. 
    If $N_S$ accepts $z$, it must be that $z = (x, p)$ for some $x \in L$ and some accepting path $p$ of $M$ on $x$, and that $M$ on input $x$ has an accepting path $p'$ that is lexicographically smaller than $p$, i.e., $z \in S$.
    If $z\in S$, then $z$ is a pair, let us call it $(x, p)$. The machine $N_S$ verifies that $x \in L$ using $p$, and nondeterministically guesses a path $p'$ of $M$ on $x$ that is lexicographically smaller than $p$. Since $(x, p) \in S$, there is such a path $p'$ that is accepting, and so $N_S$ accepts.

    To conclude, it is clear that $N_S$ runs in nondeterministic polynomial time and that for any $z=(x, p) \in S$, 
    $\acc_{N_S}(z) = \acc_M(x) - 1$, 
    because the (lexicographically) smallest accepting path of $M$ on $x$ is ``poisoned.'' Equivalently, this means that $\acc_{N_S}(z) \leq f(|x|) + j$. Next, because $f$ is monotonically nondecreasing and $|x| \leq |z|$, we obtain that $\acc_{N_S}(z) \leq f(|z|) + j$. Therefore $S \in \upleq{f(n)+j} = \p$ (by our inductive hypothesis).
    \end{proof}

    \begin{lemma}\label{l:r-in-p}
        $R\in\p$.
    \end{lemma}
    \begin{proof}
    Let $N_R$ be an NPTM that does the following on input $x$. 
    Guess a path $p$ of $M$ on $x$. If $p$ is not an accepting path, reject.
    If $(x, p) \in S$, then accept. Otherwise, reject.

    Because $S \in \p$, the membership test of $(x, p)$ in $S$ can be computed in polynomial time. So $N_R$ clearly runs in nondeterministic polynomial time. Also, for each input $x\in L$, there are only up to $f(|x|) + j$ accepting paths, and $N_R$ accepts if and only if $M$ has more than one accepting path on $x$, so it follows that $L(N_R) = R \in \upleq{f(n)+j} = \p$ (by our inductive hypothesis).
    \end{proof}

    Now to conclude, consider the NPTM $N_L$ that does the following on input $x$.
    If $x \in R$ (which can be computed in polynomial time by Lemma~\ref{l:r-in-p}), accept. Otherwise, immersively simulate $M$ on input $x$.

    $N_L$ clearly runs in nondeterministic polynomial time, accepts $L$, and has at most one accepting path on any input. So, $L \in \up=\p$.
\end{proof}

Our oracle from Theorem~\ref{t:oracle-super-constant} gives us an indication that any improvement to Theorem~\ref{t:constant-addition}
(taking the case where $f(n)$ is constant bounded) 
cannot hold robustly.
However, 
we can combine Theorems~\ref{t:poly-stretch} and~\ref{t:constant-addition} to get the following result.

\begin{restatable}{theorem}{positivestrengthening}
\label{t:positive-strengthening}
      For each 
      monotonically nondecreasing
      function $\ntor{f}$
      and each $k \in \naturalnumberpositive$, $\p = \upleq{f(n)} \iff \p = \upleq{f(n^k)+\bigoh(1)}$.
\end{restatable}

As a corollary to Theorem~\ref{t:positive-strengthening}, we obtain the following concrete examples of linked fates.

\begin{corollary}\label{c:concrete-examples}~
  \begin{enumerate}
  \item 
    $\p = \upleq{\log\log(n)} \iff \p =
    \upleq{\bigoh(1) + \log\log(n)}$.\footnote{Though it is
      uglier, we can even from the previous theorem
      draw the slightly stronger conclusion that 
for each $j \in \naturalnumber$, it holds that
    $\p = \upleq{\max(\log\log(n) - j\,,\, 1)} \implies \p =
    \upleq{\bigoh(1) + \log\log(n)}$.}

  \item $\p = \upleq{\log(n)} \iff \p = \upbigohlogn$.

\item 
  For each $k \in \naturalnumberpositive$,
    $\p = \upleq{n^k + k} \iff \p = \fewp$.

  \item \label{i:nlogn-collapse} 
  For each $j\in\naturalnumber$,
  $\p = \upleq{j+ n^{\log(n)}} \iff \p = 
  \upleq{n^{\bigoh(\log(n))}}$.

  \item \label{i:exp-collapse} 
  For each $k \in \naturalnumberpositive$,
  $\p = \upleq{2^{n^k}} \iff \p = \np$.
  \end{enumerate}
\end{corollary}

We mention that, as one can easily see from its proof,
Theorem~\ref{t:padding} relativizes. 
It immediately 
follows that Theorem~\ref{t:poly-stretch} and 
Corollary~\ref{c:concrete-examples} relativize.
Additionally, we can also observe that the proof that Theorem~\ref{t:constant-addition} relativizes, 
so it is clear that the proof of Theorem~\ref{t:positive-strengthening} also relativizes.

Theorem~\ref{t:poly-stretch} shows that polynomially-bounded
increases in the argument of the ambiguity limit yield linked-fate
pairs of classes, and the corollaries just given provide a number of
concrete examples. In a similar fashion, Theorem~\ref{t:constant-addition} shows that an addition of a constant amount of ambiguity also yields linked-fate pairs of classes.
And so, just as Theorem~\ref{t:oracle-super-constant}
shows that Theorem~\ref{t:osamu} could not be stretched any further by
relativizable techniques, 
so too can we 
show that
Theorem~\ref{t:positive-strengthening} would have relativized counterexamples if
one were to change the increase in the argument to more than a
polynomial, or if one were to add an unbounded function to the ambiguity bound.

\begin{theorem}\label{t:stretch-oracle}
    Let $\ntor{f_1}$ and 
	$\ntor{f_2}$ 
	be any pair of 
    monotonically nondecreasing
    functions such that
    (a)~$f_2(n) = 2^{n^{\bigoh(1)}}$,
    (b)~for each $n \in \naturalnumber$, $f_1(n) \leq f_2(n)$,%
    \footnote{\label{f:level-and-surge}
    With
 some 
	special-case work on a finite prefix, we could prove the theorem 
	even under the slightly weaker hypothesis ``for all but at most a 
	finite number of $n \in \naturalnumber$, 
	$f_1(n) \leq f_2(n)$.''  But that is not a 
	particularly interesting ``strengthening,'' and so we will not do 
	it.} and 
    (c)~for no $j\in\naturalnumber$ and no $k \in \naturalnumberpositive$ does it hold that, for all but a finite number of $n \in \naturalnumber$, $f_2(n) \leq f_1(n^k)+j$.
	Then there exists an 
    oracle $B$ such that $\p^B =
	\upleq{f_1(n)}^B \subsetneq \upleq{f_2(n)}^B$.
    Moreover, if the functions $\floor{f_1(n)}$ and $\floor{f_2(n)}$ are both computable, then $B$ is recursive.
\end{theorem}

We defer the proof of Theorem~\ref{t:stretch-oracle} to the next section, but before presenting it, we justify the conditions in that theorem's statement.

Conditions~(a) and~(c) ward off the case where the two functions are
so large that the theorem is made inherently impossible to prove due
to having, for all $A$, $\upleq{f_1(n)}^A = \upleq{f_2(n)}^A = \np^A$---via our Theorem~\ref{t:positive-strengthening}.
Indeed, both conditions imply that
$f_1(n) \neq 2^{n^{\Omega(1)}}$,
and so $\upleq{f_1(n)}$ would not have enough ambiguity to, in any
obvious way, equal $\np$.  (We are not saying that that class isn't
equal to $\np$ in some worlds.  We are just explaining why the
claim/worry ``just use two huge functions $f_1$ and $f_2$ since then
both classes will robustly be $\np$ and so they are equal and can't be
oracle-separated'' isn't a convincing claim.)

We also do not have to worry about the case where $f_1$ and $f_2$ are both constant-bounded, which by Theorem~\ref{t:osamu} would imply that $\p = \upleq{f_1(n)} = \upleq{f_2(n)}$ not just in the real world, but also in every relativized world. Moreover, if $f_2$ were bounded by a constant $j \in \naturalnumber$, then so would $f_1$ by condition~(b), which would imply that for all $n \in \naturalnumber$, $f_2(n) \leq f_1(n) + j$, thereby violating condition~(c).

\section{Proof of Theorem~\ref{t:stretch-oracle}}\label{s:big-proof}

    Let $f_1$ and $f_2$ be functions satisfying the conditions in the statement of Theorem~\ref{t:stretch-oracle}.
    Let $H$ be a $\pspace$-complete set whose alphabet is $\Sigma=\{0,1\}$. The disjoint union operator is defined by $A \oplus B = \{0x \mid x \in A\} \cup \{1y \mid y \in B\}$. For concreteness, we will take $\oplus$ to be right associative.
    In this proof we construct an oracle $B = H \oplus C \oplus E$. Moreover, $H$, $C$, and $E$ will each be over the alphabet $\Sigma$. 

    We let $\sigmastar$ denote the set of all strings over $\Sigma$. For each $n \in \naturalnumber$, let $\Sigma^{=n} = \{x \in \sigmastar \mid |x| = n\}$ and let $\Sigma^{\leq n} = \{x \in \sigmastar \mid |x| \leq n\}$.

    Because there are many ideas being used in this proof, in Section~\ref{s:groundwork} we first give 
    a brief 
    overview 
    of 
    our proof.
    We then introduce
    some useful lemmas.

    \subsection{%
    Groundwork}\label{s:groundwork}

    The $M_i$'s and $N_i$'s of Section~\ref{s:big-proof} always refer to the machines 
    of the standard enumerations
    mentioned from Section~\ref{s:defs}.
    $E$ will be used to separate $\p^B$ and $\upleq{f_2(n)}^B$ by guaranteeing that the language
    $L_E = \{0^{n} \mid n\geq 0 \land (\exists y \in E)[|y| = n]\}$ is in $\upleq{f_2(n)}^B - \p^B.$
    To do so, we will 
    construct $L_E$ by diagonalizing against every 
    $M_i$, $i\geq 1$,
    (recall that those machines have the property that for each $i$ and each $A$, the running time of $M^A_i$ on inputs of length $n$ is at most $n^i+i$) with oracle $B$.
    This will be done while ensuring that we never add to $E$ more than $f_2(n)$ strings at each length $n$.
    Doing so will guarantee that $L_E \in \upleq{f_2(n)}^B$.
    Moreover, we will allow $E$ to contain only strings with 
    tower lengths to help us ensure that later stages in our 
    construction do not affect earlier stages.
    
    We will 
    try to make $N_i$'s having oracle $B$ not be $f_1$-categorical,
    and we will also try to make pairs of $(N_i, N_j)$'s having oracle $B$ not accept complementary languages,
    which will help us show that $\p^B = \np^B\cap\conp^B$---using $H$ and a result of Hartmanis and Hemachandra~\cite{har-hem:j:rob}, namely Lemma~\ref{l:har-hem} below. %

    \begin{lemma}[\cite{har-hem:j:rob}]\label{l:har-hem}
        Assume $\p=\np$. Let $N_i$ and $N_j$ robustly accept complementary languages. Then for every oracle $O$, $L(N_i^O) \in \p^O$.
    \end{lemma}

    \begin{corollary}\label{c:p-np-conp}
        Let $A$ be an oracle such that $\p^A = \np^A$, and let $N_i$ and $N_j$ robustly accept complementary languages. Then for every oracle $O$, $L(N_i^{A \oplus O}) \in \p^{A \oplus O}$.
    \end{corollary}
    \begin{proof}
        Since the proof of Lemma~\ref{l:har-hem} given in \cite{har-hem:j:rob} relativizes and the assumption of the lemma holds relative to $A$ (i.e., $\p^A=\np^A$), it follows (unconditionally) that for 
        every $N_i$ and $N_j$
        that robustly accept complementary languages, and for every oracle $O$, $L(N_i^{A \oplus O}) \in \p^{A \oplus O}$.
    \end{proof}
    Finally, 
    for each $L \in \coup_{\leq f_1(n)}^B$,
    we will
    add strings in $C$ that will be used to determine membership in $L$ in nondeterministic polynomial time relative to oracle $B$\@. 
    This will allow us to establish that
    $\coup_{\leq f_1(n)}^B \subseteq \np^B$. By definition, $\coup_{\leq f_1(n)}^B \subseteq \conp^B$, so it will immediately follow that $\p^B = \upleq{f_1(n)}^B = \np^B \cap \conp^B$.

    We will use the next lemma. %

    \begin{lemma}\label{l:inf-oft-exp}
        For each $\epsilon > 0$, $f_1(n) < 2^{\epsilon n}$ for infinitely many values of $n$.
    \end{lemma}
    \begin{proof}
    
    Let $\epsilon > 0$ and suppose for the sake of contradiction that $f_1(n) \geq 2^{\epsilon n}$ for all but finitely many values of $n$. In other words, there is a positive constant $\hat{n}_0$ such that for all $n \geq \hat{n}_0$, $f_1(n) \geq 2^{\epsilon n}$. Let $\hat{n}_0$ be such a constant.
    
    Moreover, recall from condition~(a) of Theorem~\ref{t:stretch-oracle} that $f_2(n) = 2^{n^{\bigoh(1)}}$.
    Therefore, there are positive integers $c$ and $n_0$ such that for all $n \geq n_0$, we have $f_2(n) \leq 2^{n^c}$. Let $c$ and $n_0$ be such positive integers.
    Additionally, let $n_0' = \max\{\ceiling{1/\epsilon}, n_0, \hat{n}_0\}$. 
    
    Notice that for all $n \geq n_0'$, it holds that $\epsilon n \geq 1$. So it follows using the assumption in the first sentence of this proof and the inequality in the prior sentence---along with the fact that the functions below are monotonically nondecreasing---that for all $n \geq n_0'$, %
        $f_1(n^{c+1}) \geq 2^{\epsilon n^{c+1}} = 2^{\epsilon n\cdot n^{c}} \geq 2^{n^c} \geq f_2(n),$
    which violates condition~(c) of Theorem~\ref{t:stretch-oracle}.
    \end{proof}

    \subsection{%
    Oracle Construction}

    We define our oracle using a 
    stage
    construction:
    Our oracle is the outcome of an infinite sequence of stages $0, 1, 2, \ldots$, each of which may add some strings to the $C$ or $E$ parts of the oracle.
    The proof of correctness for the oracle is deferred to the appendix. 

    For each $n \in \naturalnumber$, by the end of stage $n$ we will have defined $C(n)$ (resp. $E(n)$), which will be the contents of $C$ (resp. $E$) at the end of stage $n$.
    Strings added to the oracle during a particular stage will never be removed in a later stage, so it follows that for each stage $n$, $C(n) \subseteq C(n+1)$ and $E(n) \subseteq E(n+1)$. So, clearly, $C = \bigcup_{n \in \naturalnumber} C(n)$ and $E = \bigcup_{n \in \naturalnumber} E(n)$ (and recall that $B = H \oplus C \oplus E$).
    Moreover,
    for each stage $n$, we let $B(n) = H \oplus C(n) \oplus E(n)$, for brevity.
    During the proof, we will declare certain strings to be ``frozen'' out of $C$. That means they are not yet at that part of the construction members of $C$, and we commit to not add them to $C$ in the rest of the construction. They are thus guaranteed to be in $\overline{C}$. We use the notion of freezing strings out analogously for $E$. 
    In  particular, since we build $C$ and $E$ in stages, let $C'(n)$ (resp. $E'(n)$) denote the strings we by the end of stage $n$ declared to be frozen out of $C$ (resp. $E$).

    Let $p_i(n) = n^i + i$. We define below 
    requirements that we will attempt to make our oracle satisfy. However, not all the requirements will be necessarily 
    satisfied. For each $i\in \naturalnumberpositive$,
    \begin{enumerate}
        \item let $R_i$ denote the requirement that ``$L(M_i^B) \neq L_E$,''
        \item let $S_i$ denote the requirement that ``$(\exists x)[\acc_{N_i^B}(x) \geq \floor{f_1(|x|)} + 1]$,'' and
        \item for each $j$, $1\leq j < i$, let $T_{i,j}$ denote the requirement that ``$L(N_i^B) \neq \overline{L(N_j^B)}$.''
    \end{enumerate}

    For each $n \in \naturalnumberpositive$, we will call stage $n$ an ``odd stage'' if $n$ is odd, and we will call stage $n$ an ``even stage'' if $n$ is even.
    While we do not explicitly give requirements for 
    adding to $C$ those strings that will be used to determine membership in $\coup_{\leq f_1(n)}^B$ languages,
    but we will make it clear (as part of the proof of Claim~\ref{claim:coupf1-in-np}) that that task is achieved by the ``odd'' stages of the construction.
    Moreover, we will only add strings to $C$ during odd stages and we will only add strings to $E$ during even stages.

    Whenever we speak of a string $x$ being queried to $X \in \{H, C, E\}$, we assume that the prefix of the query string that determines which part of the oracle is being queried is not included in $x$.

    We now give for each of the requirements defined above the conditions that must hold for the requirement to be satisfied at a given stage $n$ and how to satisfy that requirement in that same stage $n$. We give these descriptions in Sections~\ref{subsub:satisfy-r}--\ref{subsub:satisfy-t} and then give the stage construction in Section~\ref{subsub:stage-construction}.

    \subsubsection
    {Satisfying \boldmath$R_i$ Requirements}\label{subsub:satisfy-r}
    By making our oracle $B$ satisfy each of the $R_i$ requirements, we will have that $\p^B \neq \upleq{f_2(n)}^B$, as we will guarantee that $L_E \in \upleq{f_2(n)}^B$ and it holds that $\p^B = \{L(M_i^B)\}_{i\in\naturalnumberpositive}$.

    \begin{algorithm}[h]
    \caption{Algorithm to satisfy the $R_i$ requirements.}\label{alg:satisfy-r}
    \begin{algorithmic}[1]
        \Require $n$ is a tower length.
        \Require $C(n-1)$, $E(n-1)$, $C'(n-1)$, and $E'(n-1)$ are subsets of $\Sigma^{\leq n-1}$.
        \Require $p_i(n) + \floor{f_1(n)} < 2^{n/2}$.
        \Procedure{Satisfy-$R$}{$i$, $n$}
        \State Simulate $M_i^{B(n-1)}$ on input $0^n$.
        \State Let $Q_E$ denote the strings queried to $E$ that received answer ``no'' and that are not currently frozen out of $E$.
        \State Let $Q_C$ denote the strings queried to $C$ that received answer ``no'' and that are not currently frozen out of $C$.
        \If{$M_i^{B(n-1)}$ rejects on input $0^n$}
            \State Let $N_E$ denote the strings of length $n$ that are currently \emph{not} frozen out of $E$.
            \State\label{alg-r:add-strings} 
            Let $Y$ be the $\min\{\|N_E\|, \floor{f_2(n)}\}$ lexicographically smallest strings in $N_E$.
            \State Let $E(n) = E(n-1) \cup Y$.
        \EndIf
        \State Let $C(n) = C(n-1)$.
        \State Freeze out of $C$ in string in $Q_C$.
        \State Freeze out of $E$ in string in $Q_E$.
        \EndProcedure
    \end{algorithmic}
    \end{algorithm}

    The procedure \textsc{Satisfy-$R$} adds strings of length $n$ to $E$ and we only add strings with tower lengths to $E$, thus explaining the first precondition of the procedure. Moreover, the second precondition of the procedure ensures that steps taken in prior stages do not affect the steps that will be taken during the execution of the procedure. Finally, the third precondition of the procedure ensures that $M_i^{B(n-1)}$ on input $0^{n}$ cannot 
    query every string of length $n$
    and that the number of strings of length $n$ in $C'(n)-C'(n-1)$, which are determined during the execution of \textsc{Satisfy}-$R$, is less than $2^{n/2}$ (this will become relevant later, in the proof of Claim~\ref{claim:coupf1-in-np}). Moreover, this same precondition of the procedure guarantees that on line~\ref{alg-r:add-strings}, at least $\floor{f_1(n)} + 1$ strings will be added to $E$. 
    
    To see why the procedure works, notice that
    if $M_i^{B(n-1)}$ accepts $0^n$, then $0^n \not\in L_E$, and if $M_i^{B(n-1)}$ rejects $0^n$, then $0^n \in L_E$. 
    Moreover, in our construction, once a string is added to $B$, it is never removed from $B$, and as we mentioned earlier, once a string is frozen out of $E$ during some stage, it is guaranteed to be in $\overline{E}$. Therefore, 
    $0^n \in L(M_i^{B(n-1)}) \iff 0^n \in L(M_i^B)$. 
    In sum, if $R_i$ is satisfied, then
    $L_E \neq L(M_i^B)$.

    \subsubsection
    {Satisfying \boldmath$S_i$ Requirements}\label{subsub:satisfy-s}
    
    Our $S_i$ requirements seek to make each $N_i^B$ have at least $f_1(n)+1$ accepting paths on some input of length $n$, thus guaranteeing that $N_i^B$ is not $f_1$-categorical. However, the construction cannot guarantee our oracle $B$ satisfies each $S_i$ requirement. This is however not an issue as doing so will allow us to draw a connection, as part of the proof of Claim~\ref{claim:robust-equality}, between unsatisfied $S_i$ requirements and robustly $f_1$-categorical oracle machines in the enumeration $N_1$, $N_2$, $N_3$, $\ldots$ from Section~\ref{s:defs}. Doing so will be helpful in proving both Claims~\ref{claim:robust-equality} and~\ref{claim:coupf1-in-np}, which will ultimately underpin the proof that $\p^B = \upleq{f_1(n)}^B$.

        \begin{algorithm}[h]
    \caption{Algorithm to satisfy the $S_i$ requirements.}\label{alg:satisfy-s}
    \begin{algorithmic}[1]
        \Require $n$ is a tower length.
        \Require $C(n-1)$, $E(n-1)$, $C'(n-1)$, and $E'(n-1)$ are subsets of $\Sigma^{\leq n-1}$.
        \Require $p_i(n)(\floor{f_1(n)}+1) < 2^{n/2}$.
        \Require There is a set $U_E \subseteq \Sigma^{=n}$, where $\|U_E\| \leq f_2(n)$, and a string $x \in \Sigma^{\leq n}$, such that $N_i^{H \oplus C(n-1) \oplus (E(n-1) \cup U_E)}$ has at least $\floor{f_1(|x|)}+1$ accepting paths on input $x$.
        \Procedure{Satisfy-$S$}{$i$, $n$}
            \State Let $U_E$ and $x$ satisfy the last precondition of the procedure.
            \State Simulate $N_i^{H \oplus C(n-1) \oplus (E(n-1) \cup U_E)}$ on $x$.
            \State Let $Q_E$ denote the strings that were queried to $E(n-1)\cup U_E$ in the above simulation, along one of the $\floor{f_1(|x|)}+1$ lexicographically smallest accepting paths, that received answer ``no'' and that are not currently frozen out of $E$.
            \State Let $Q_C$ denote the strings that were queried to $C(n-1)$ in the above simulation, along one of the $\floor{f_1(|x|)}+1$ lexicographically smallest accepting paths, that received answer ``no'' and that are not currently frozen out of $C$.
            \State Let $E(n) = E(n-1) \cup U_E$.
            \State Let $C(n) = C(n-1)$.
            \State Freeze out of $C$ each string in $Q_C$. %
            \State Freeze out of $E$ each string in $Q_E$. %
        \EndProcedure
    \end{algorithmic}
    \end{algorithm}

    The justifications for the first two preconditions for \textsc{Satisfy}-$S$ are the same as the ones for \textsc{Satisfy}-$R$. The third precondition guarantees that $Q_E$ and $Q_C$ each contain fewer than $2^{n/2}$ strings, which will become relevant in the proof of Claim~\ref{claim:coupf1-in-np}. 
    If the fourth precondition is satisfied, then there is a way to make our oracle satisfy requirement $S_i$, by doing exactly what is done in \textsc{Satisfy}-$S$. 
    We also restrict $Q_E$ and $Q_C$ to be about queries done along the $\floor{f_1(|x|)}+1$ lexicographically smallest accepting paths---which are guaranteed to exist by the fourth precondition---in order to prove Claim~\ref{claim:coupf1-in-np}. 

    By using similar arguments as the ones used to justify the correctness of \textsc{Satisfy}-$R$, it follows that if $S_i$ is satisfied by our oracle $B$, then $N_i^B$ is not $f_1$-categorical.

    \subsubsection{Satisfying \boldmath$T_{i,j}$ Requirements}\label{subsub:satisfy-t}

    Our $T_{i,j}$ requirements seek to prevent pairs $(N_i^B, N_j^B)$ from accepting complementary languages whenever possible. 
    Just like with the $S_i$ requirements, the construction we give cannot guarantee that our oracle $B$ satisfies each $T_{i,j}$ requirement, but in the proof of Claim~\ref{claim:coupf1-in-np}, we will draw useful information from unsatisfied $T_{i,j}$ requirements that will allow us to prove that $\p^B = \np^B \cap \conp^B$, thereby helping us leverage the steps taken in the odd stages to prove that $\p^B = \upleq{f_1(n)}^B$.\footnote{If there were a polynomial $q(n)$ such that for all but a finitely many values of $n$, $f_1(n) \leq q(n)$, it would suffice to leverage only the $R_i$ and $S_i$ requirements to construct our desired oracle, e.g., using the approach used by Fenner, Fortnow, and Kurtz~\cite{fen-for-kur:jisomorphism-holds}. However, we make no such assumption on $f_1$, and thus use a two-step process to obtain our result that $\p^B = \upleq{f_1(n)}^B$.}

    \begin{algorithm}[h]
    \caption{Algorithm to satisfy the $T_{i,j}$ requirements.}\label{alg:satisfy-t}
    \begin{algorithmic}[1]
        \Require $n$ is a tower length.
        \Require $C(n-1)$, $E(n-1)$, $C'(n-1)$, and $E'(n-1)$ are subsets of $\Sigma^{\leq n-1}$.
        \Require $p_i(n) + p_j(n) < 2^{n/2}$.
        \Require There is a set $U_E \subseteq \Sigma^{=n}$, where $\|U_E\| \leq f_2(n)$, and a string $x \in \Sigma^{\leq n}$, such that $N_i^{H \oplus C(n-1) \oplus (E(n-1) \cup U_E)}$ and $N_j^{H \oplus C(n-1) \oplus (E(n-1) \cup U_E)}$ 
         both accept $x$.
        \Procedure{Satisfy-$T$}{$i$, $j$, $n$}
            \State Let $U_E$ and $x$ satisfy the last precondition of the procedure.
            \State Simulate $N_i^{H \oplus C(n-1) \oplus (E(n-1) \cup U_E)}$ on $x$.
            \State Simulate $N_j^{H \oplus C(n-1) \oplus (E(n-1) \cup U_E)}$ on $x$.
            \State Let $Q_E$ denote the strings that were queried to $E(n-1)\cup U_E$ in the above simulations, along the lexicographically smallest accepting path of each machine, that received answer ``no'' and that are not currently frozen out of $E$.
            \State Let $Q_C$ denote the strings that were queried to $C(n-1)$ in the above simulations, along the lexicographically smallest accepting path of each machine, that received answer ``no'' and that are not currently frozen out of $C$.
            \State Let $E(n) = E(n-1) \cup U_E$.
            \State Let $C(n) = C(n-1)$.
            \State Freeze out of $C$ each string in $Q_C$.
            \State Freeze out of $E$ each string in $Q_E$.
        \EndProcedure
    \end{algorithmic}
    \end{algorithm}

    The justifications for the first three preconditions for \textsc{Satisfy}-$T$ are the same as the ones for \textsc{Satisfy}-$S$.
    If the fourth precondition is satisfied, then there is a way to make our oracle satisfy requirement $T_{i,j}$, by doing exactly what is done in \textsc{Satisfy}-$T$. 
    We also restrict $Q_E$ and $Q_C$ to be about queries done along the lexicographically smallest accepting paths---which is guaranteed to exist by the fourth precondition---in order to prove Claim~\ref{claim:coupf1-in-np}. 

    By using similar arguments as the ones used to justify the correctness of \textsc{Satisfy}-$R$, it follows that if $T_{i,j}$ is satisfied by our oracle $B$, then $N_i^B$ and $N_j^B$ do not accept complementary languages. And by running \textsc{Satisfy}-$T$, we can guarantee, whenever possible, to make $N_i^B$ and $N_j^B$ accept a common string. The construction does not directly handle the case where both reject a common string (this is because there is no obvious way in that case to guarantee that $Q_E$ and $Q_C$ each have at most $2^{n/2}$ strings). However, this does not affect our ability to prove Claim~\ref{claim:robust-equality}, which is the main reason we have the $T_{i,j}$ requirements.

    \subsubsection{Stage Construction}\label{subsub:stage-construction}

    Having described how we satisfy the $R_i$/$S_i$/$T_{i,j}$ requirements in our construction, we are now ready to give the complete stage construction.

    \noindent
    \textbf{Begin stage descriptions.}
    \begin{enumerate}
    \item \textbf{Stage 0:} Let $C(0) = E(0) = \emptyset$, and thus $B(0) = H \oplus C(0) \oplus E(0)$. %

    \item \textbf{Stage $\bm{n = 2m+1, m\in \naturalnumber}$
    (aka ``an odd stage''):}
    \begin{enumerate}
        \item We define the set $Y$ as follows: for each $y \in \Sigma^{=n}-C'(n-1)$, $y \in Y$ if and only if there is a $j\in \naturalnumberpositive$ and a string $x \in \sigmastar$ such that $(j, x, 1^{p_j(|x|)})$ is encoded in the $m$-character prefix of $y$,\footnote{We tacitly assume here that there is a polynomial-time computable and polynomial-time invertible encoding that maps tuples $(w, y, z)$, where $w \in \naturalnumber$ and $y,z$ are 
        strings, to a string $s$ over the alphabet of $C$ such that $|s| > |w| + |y|+|z|$. Such an encoding clearly exists.}
        $S_j$ is not satisfied, and $x \not\in L(N_j^{B(n-1)})$.\\
        \underline{Note:} We will later, in the proof of Claim~\ref{claim:coupf1-in-np}, argue that $\Sigma^{=n}-C'(n-1)$ is nonempty.
        \item Add each string in $Y$ to $C$.\\
        \underline{Note:} Let $y \in Y$ and let $(j, x, 1^{p_j(|x|)})$ be the tuple encoded in the $m$-character prefix of $y$ as described in the definition of $Y$.
        Notice that $N_j^{H \oplus (C(n-1) \cup Y) \oplus E(n-1)}$ on input $x$ cannot query $y$ to $C(n-1) \cup Y$, as $y$ is too long.
        Moreover, we will never add strings of length at most $n$ to $C$ in future stages.
        Therefore $x \not\in L(N_j^B) \iff x\not\in L(N_j^{H \oplus (C(n-1) \cup Y) \cup E(n-1)})$.\footnote{We are not precluding the possibility that after adding such a $y$ at stage $n$, we satisfy $S_j$ at a later stage.  But that is fine as in such a case, we would only be adding a finite number of ``solutions'' for $N_j^B$ into $C$.}
        \item Freeze out of $C$ each string of length $n$  that is not currently in $C$ (this prevents the addition of length-$n$ strings to $C$ in future stages), and proceed to stage $n+1$.
    \end{enumerate}

    \item \textbf{Stage $\bm{n = 2m, m\in\naturalnumberpositive}$, (aka ``an even stage'')}:
    \begin{enumerate}
    \item If $n$ is not a tower length, 
    freeze all strings of length $n$ out of $E$ and 
    proceed to stage $n+1$. 
    \item Otherwise, let $n = e(k)$. 
    \item If $k \equiv 0 \pmod 3$ and there is a positive integer $i \leq n$ such that $R_i$ has not yet been satisfied, but $i$ and $n$ satisfy the preconditions of the procedure \textsc{Satisfy}-$R$ (defined in Algorithm~\ref{alg:satisfy-r}), then run \textsc{Satisfy}-$R$.

    \item If $k \equiv 1 \pmod 3$ and there is a positive integer $i \leq n$ such that $S_i$ has not yet been satisfied, 
    but $i$ and $n$ satisfy the preconditions of the procedure \textsc{Satisfy}-$S$ (defined in Algorithm~\ref{alg:satisfy-s}), then run \textsc{Satisfy}-$S$.
    
    \item If $k \equiv 2 \pmod 3$ and there are positive integers $i$ and $j$, $j < i \leq n$, such that $T_{i,j}$ has not yet been satisfied, 
    but $i$, $j$, and $n$ satisfy the preconditions of the procedure \textsc{Satisfy}-$T$ (defined in Algorithm~\ref{alg:satisfy-t}), then run \textsc{Satisfy}-$T$.
    \item Freeze out of $E$ each string of length $n$ that is not currently in $E$ (this prevents the addition of length-$n$ strings to $E$ in future stages), and proceed to stage $n+1$.
    \end{enumerate}
    \end{enumerate}
    \noindent
    \textbf{End stage descriptions.}

    As mentioned earlier, the proof of correctness for the oracle is deferred to the appendix.

\section{Related Work}\label{s:related}

It is important not to conflate our issue of (un)linked fates of two
classes---whether they stand or fall together as to collapsing to $\p$---with the different issue of whether the two classes can be
separated in some relativized world.
A clear example of this 
is the constant-ambiguity case mentioned in Section~\ref{s:results}.
Watanabe~\cite{wat:j:hardness-one-way} proved 
that if $\p=\upleq{1}$ then $\p=\upleq{2}$.  And that
implication holds in the
real world and every relativized world.  Nonetheless, there are
oracles $A$ for which $\upleq{1}^A \neq \upleq{2}^A$ holds, and indeed that
separation holds with probability one relative to a random
oracle~(\cite{bei:c:up1}, the probability-one proof there is flawed,
but a corrected proof has been obtained~\cite{bei:perscomm:dec93});
however, again by Watanabe's relativizable 
result, relative to that oracle $A$, it
certainly cannot hold that $\p^A = \upleq{1}^A$.

Historically, the first person to show a pair of classes whose
ambiguity levels were distant enough that in some relativized world
their fates were not linked was Rackoff, who in the early 1980s showed
that there is an oracle $B$ relative to which (and recall that
$\np = \upleq{2^{n^{\bigo(1)}}}$)
$\p^B = \up^B \neq \np^B$~\cite{rac:j:rel}. 
This was substantially tightened in a paper
by Fortnow and Rogers~\cite{for-rog:j:separability} (itself building on techniques of Blum and Impagliazzo~\cite{blu-imp:c:generic} and Hartmanis and Hemachandra~\cite{har-hem:j:one-way-iso,har-hem:j:rob}) which, 
although it is not explicitly stated in the paper, builds an oracle
$A$ such that
$\p^A=\up^A\neq
\fewp^A$~\cite{for:perscomm:p-eq-up-neq-fewp}.  
In fact, Fortnow and Rogers effectively show that relative to a ``size-bounded generic oracle''
$A$, $\p^A=\up^A\neq \upleq{n}^A$, but as we have shown in our Corollary~\ref{c:concrete-examples}, $\upleq{n}$ and $\fewp$ have linked fates, not just in the real world, but also in every relativized world.

Our main oracle result is proven by carefully taking the
methods
of Baker, Gill, and Solovay~\cite{bak-gil-sol:j:rel}, Blum and Impagliazzo~\cite{blu-imp:c:generic}, and Hartmanis and Hemachandra~\cite{har-hem:j:up,har-hem:j:rob}
as far as they can go without running
hopelessly afoul of Watanabe's linked-fates result and without fatally interfering with each other. Moreover, Theorem~\ref{t:constant-addition} proves a stronger result than that of Watanabe~\cite{wat:j:hardness-one-way} by using a different, two-step approach to ``poison'' an additional path even within superconstant ambiguity settings.

\section{Conclusion and Open Problems}

In this paper, we finally close the open problem of strengthening beyond constant-ambiguity versions of $\np$ Watanabe's 1988~\cite{wat:j:hardness-one-way} linked-fates
result.
We establish as Theorem~\ref{t:positive-strengthening} one new such class of cases.
We also prove
that any strengthening of our work cannot hold robustly, i.e., in every relativized world. 

Our positive result is proven using (a)~a path-poisoning approach that is not limited to the constant-ambiguity case and (b)~the power of padding, in particular to raise the ambiguity bound from $f(n)$ to $f(n^k)$ for monotonically nondecreasing functions $f$. Indeed, Theorem~\ref{t:positive-strengthening} draws on both (a) and (b), via Theorems~\ref{t:constant-addition} and~\ref{t:poly-stretch}, respectively. Aside from these methods that we found, it would be interesting to see if there are other techniques that can create linked fates and if they can be combined with our methods.

While our oracle results are interesting in their own right, we of course would prefer proving that structural collapses, e.g., of the polynomial hierarchy, follow from our assumptions. However, this is an area that has long proven difficult for the field, since even certain dramatic assumptions, such as $\p=\up$ or even $\up = \np$, are not known to imply the collapse of the polynomial hierarchy. 
Indeed, regarding the $\p=\up$ assumption, Blum and
Impagliazzo~\cite{blu-imp:c:generic} showed
that there is a relativized
world in which $\p = \np\cap\conp = \up$ yet the polynomial hierarchy
is infinite.  
However, there is no known relativized world where $\up=\np$ and the polynomial hierarchy is infinite~\cite{hem-ram-zim:j:worlds-to-die-for,for:j:die-harder}.
And it is open whether there
is an oracle world in which $\up$ is in the high hierarchy and $\ph$ is
infinite---in fact, Sheu and Long~\cite[p.~425]{lon-she:j:up-low-high} explicitly
mentioned this open problem in~1996, and to the best of our knowledge
it remains open even now, decades later. 
\bibliographystyle{plainurl}

\clearpage
\appendix 

\section{Deferred Note and Proofs}

\subsection{Deferred Note}

\begin{appendixnote}\label{app:encoding}\rm
Under most natural encodings, $\sat$ minus all its
strings whose lengths are perfect squares is such a set.
Or, to give a different, and more general and formal approach 
to giving such a set $B$: There are \np-complete 
sets over the alphabet $\{0,1\}$. For any such set, call it $A$,
we can take $B$ to be 
the set that at each perfect-square length contains no strings,
and that at each other length, $j$, (a)~for each $1 \leq i \leq 2^j-1$, $B$
contains the lexicographically $i$th string of length $j$ exactly 
if the lexicographically $i$th string in $\{0,1\}^*$ belongs to $A$;
and (b)~$1^j\not\in B$. That is, we use the membership in $B$ of 
the first 
$2^j-1$ length-$j$ strings to transparently encode 
the membership status in $A$
of the strings of the lengths in the range 0 through $j-1$; this type 
of packing has been very useful in earlier papers~\cite{bel-gol:j:s-vs-d,fal-ogi:j:self}. It is clear that such a set $B$ is \np-complete.
\end{appendixnote}


\subsection{Deferred Proofs of Positive Results}

For the reader's convenience, we restate the corresponding theorem of each deferred proof in this subsection.

\polystretch*
\begin{proof}
  Fix $k$ and $f$ to satisfy the theorem's requirements.
  The left-to-right direction follows from Theorem~\ref{t:padding} and the fact that 
  $\p$ is closed downwards under polynomial-time many-one
  reductions.
  Because $f$ is monotonically nondecreasing, it follows that $(\forall n)[f(n) \leq f(n^k)]$, and so the right-to-left direction holds by Fact~\ref{f:trivial}.
\end{proof}

\positivestrengthening*
\begin{proof}
Fix $f$ and $k$
to satisfy the requirements of the theorem statement.

For the right-to-left direction, for every function $j(n)$ that is $\bigoh(1)$, $j(n)$ is positive by the definition of big-Oh (see \cite[Chapter~7]{sip:b:introduction-third-edition})
  and by the monotonicity of $f$, $(\forall n\in\naturalnumber)[f(n^k) + j(n) \geq f(n)]$. So the right-to-left direction holds by Fact~\ref{f:trivial}.

We now prove the left-to-right direction. Assume that $\p=\upleq{f(n)}$.
Since $f$ is monotonically nondecreasing, it follows that $\p = \upleq{f(n^k)}$, by Theorem~\ref{t:poly-stretch}. Moreover, since $f(n^k)$ is a function from $\naturalnumber$ to $\mathbb{R}^{\geq 1}$, it follows from Theorem~\ref{t:constant-addition} that $\p=\upleq{f(n^k) + \bigoh(1)}$.
\end{proof}

\subsection{Deferred Proof of Correctness of Oracle \boldmath$B$}
In this deferred section, we prove the correctness of the oracle $B$ constructed in the proof of Theorem~\ref{t:stretch-oracle}.

    Notice that during the even stages, the only strings that are added to the oracle are strings that have tower lengths. 
    Moreover, at the end of each even stage $n$, we freeze out of $C$/$E$ the length-$n$ strings that are not in $C$/$E$.
    Therefore, if a requirement is satisfied at stage $n$, then the satisfaction of that requirement 
    cannot be changed by future stages.
    Moreover, note that requirements that have infinitely many opportunities to be satisfied will eventually be satisfied.

    \begin{property}\label{property:p-neq-upf2}
        $\p^B \subsetneq \upleq{f_2(n)}^B$.
    \end{property}
    \begin{proof}
    For every oracle $A$, $\p^A \subseteq \upleq{f_2(n)}^A$, therefore $\p^B \subseteq \upleq{f_2(n)}^B$. Also, since we never put more than $f_2(n)$ strings of length $n$ in $E$, it follows that $L_E \in \upleq{f_2(n)}^B$.

    First observe that $\p^B = \{L(M_i^B)\}_{i\in \naturalnumberpositive}$. 
    Let $i \in \naturalnumberpositive$.
    It follows from Lemma~\ref{l:inf-oft-exp} that $p_i(n)+\floor{f_1(n)} < 2^{n/2}$ for infinitely many values of $n$, and so $R_i$ is satisfiable during infinitely many stages.
    Therefore $R_i$ will eventually be satisfied. In other words, $L(M_i^B) \neq L_E$\@.
    Since every $R_i$ requirement is satisfied, no $M_i$ with oracle $B$ accepts $L_E$, and so $L_E \not\in \p^B$.
    \end{proof}

    The rest of this proof is about proving that $\p^B = \upleq{f_1(n)}^B$. %
    Let $\calC_{f_1}^B = \{L(N_i^B) \mid N_i \text{ is robustly } f_1\text{-categorical}\}$.
    Moreover, let 
    ${\calC}_{\text{compl.}}^B = \{L(N_i^B) \mid (\exists j \in \naturalnumberpositive)(\forall A)[L(N_i^A) = \overline{L(N_j^A)}]\}.$

    \begin{claim}\label{claim:robust-equality}
        ${\calC}_{f_1}^B = \upleq{f_1(n)}^B$ and $\calC_{\rm{compl.}}^B = \np^B \cap \conp^B$.
    \end{claim}
    \begin{proof}
    Clearly ${\calC}^B_{f_1} \subseteq \upleq{f_1(n)}^B$, we focus on the proof that $\upleq{f_1(n)}^B \subseteq {\calC}^B_{f_1}$.
    Let $L \in \upleq{f_1(n)}^B$ and observe that for some $i \in \naturalnumberpositive$, $N_i^B$ is $f_1$-categorical and $N_i^B$ accepts $L$.
    Next, recall if requirement $S_i$ had had infinitely many opportunities to be satisfied in the oracle construction, then $N_i^B$ would not be $f_1$-categorical.
    But, since $N_i^B$ is $f_1$-categorical, it follows that $S_i$ only had finitely many opportunities to be satisfied in the oracle construction. If $N_i$ is robustly $f_1$-categorical, then our proof is done.
    So let us now suppose 
    that $N_i$ is not robustly $f_1$-categorical.
    Let us now inspect the conditions that must be true for $S_i$ to be satisfiable.

    It follows from Lemma~\ref{l:inf-oft-exp} that $p_i(n)\left(\floor{f_1(n)}+1\right) < 2^{n/2}$ is true infinitely often.
    Therefore, the only way for $N_i$ to only have had finitely many opportunities to be made not $f_1$-categorical is if
    there were only a finite number of strings $x$ for which there were suitable sets $\hat{C}$ and $\hat{E}$ such that $N_i^{H\oplus \hat{C} \oplus \hat{E}}$ on input $x$ has more than  $f_1(|x|)$ accepting paths.
    Let the set of such strings be $I$. We will now argue that there is an $N_\ell$ that is robustly $f_1$-categorical such that $L(N_\ell^B) = L(N_i^B)$.

    Let $\hat{N}$ be an nondeterministic oracle that does the following when its oracle is some set $A$. 
    If $x \in I$, then using a finite-sized table $\hat{N}^A$ accepts if $x \in L(N_i^B)$ and rejects otherwise. 
    If $x\not\in I$, then $\hat{N}^A$ performs the same steps as $N_i^B$ on input~$x$, and accepts/rejects accordingly.

    First notice that for each oracle $A$, $L(\hat{N}^A) = L(N_i^B)$. Also notice that for each oracle $A$ and input $x$, $\acc_{\hat{N}^A}(x) \leq \acc_{N_i^B}(x)$ and $L(\hat{N}^A) = L(N_i^B)$. Therefore $\hat{N}$ is robustly $f_1$-categorical.
    Moreover, it's not too hard to see---from the way $\hat{N}$ is defined---that there is an $N_\ell$ so that for every oracle $A$, $N_\ell^A$ and $\hat{N}^A$ effectively perform the same steps on each input. 
    Therefore $N_\ell^B$ also accepts $L$ and $N_\ell^B$ is also robustly $f_1$-categorical.
    So we can conclude that $L \in {\calC}_{f_1}^B$.

    Using a similar argument, $\np^B \cap \conp^B = \calC_{\text{compl.}}$.
    \end{proof}

    We will now show that satisfying the $S_i$ 
    requirements whenever possible helps us add useful strings to $C$ that we can leverage to establish that $\coup_{\leq f_1(n)}^B \subseteq \np^B$.

    \begin{claim}\label{claim:coupf1-in-np}
        $\coup_{\leq f_1(n)}^B \subseteq \np^B$. 
    \end{claim}
    \begin{proof}
    We prove this
    by first seeing
    why the odd stages are correctly adding to $C$ strings that can be used to test for membership for $\coup_{\leq f_1(n)}^B$ languages in nondeterministic polynomial-time relative to oracle $B$.

    At any given even stage $n'=2m'$, if $n'$ is a tower length $e(k)$, then by the conditions that allow any requirement to be satisfiable, a machine simulated in that given stage queries less than $2^{e(k)/2} = 2^{m'}$ strings.
    So at the start of an odd stage $n'' = 2m''+1$, the number of strings of length $2m''+1$ that can be in $C'(n'')$ is at most $2^{1} + \cdots + 2^{m''} < 2^{m''+ 1}$. Thus during stage $n''$, each length-$m''$ string encoding tuples of the form $(k, x, 1^{p_k(|x|)})$ is the prefix of at least one string of length $2m''+1$ that is not in $C'(n'')$.
    So if $S_k$ is never satisfied, every such tuple is added to $C$ as the prefix of some string during an odd stage.

    By Claim~\ref{claim:robust-equality}, ${\calC}_{f_1}^B = \upleq{f_1(n)}^B$, so it suffices to show that 
    if $S_k$ is never satisfied, i.e., $N_k^B$ is $f_1$-categorical, then $\overline{L(N_k^B)} \in \np^B$. Indeed, consider the following nondeterministic oracle Turing machine $N^B$---which does not need to be in our 
    standard enumeration
    $N_1, N_2, N_3, \ldots$ from Section~\ref{s:defs}---with oracle $B$ on input $x$. 
    The machine
    $N^B$ in polynomial time computes the string $y$ that encodes $(k, x, 1^{p_k(|x|)})$, 
    nondeterministically guesses a string $z$ of length $|y|+1$, and queries the oracle to ask if $yz \in C$. If the oracle responds ``yes,'' then $N_k^B$ rejects $x$ and so $N^B$ accepts on the current path. If the oracle responds ``no,'' then $N_k^B$ accepts $x$ and so $N^B$ rejects on the current path.

    Note that $N_k^B$ accepts $x$ if and only if $N^B$ has no accepting path on $x$; in other words, $x \in L(N^B) \iff x \not\in L(N_k^B)$.
    Therefore $L(N^B) = \overline{L(N_k^B)}$ and $\coup_{\leq f_1(n)}^B \subseteq \np^B$.
    \end{proof}

    Finally, taking stock of all the results we have established thus far, we establish the desired property of our oracle that $\p^B = \upleq{f_1(n)}^B$.

    \begin{property}\label{property:p-eq-upf1}
        $\p^B = \upleq{f_1(n)}^B$.
    \end{property}
    \begin{proof}
    By Corollary~\ref{c:p-np-conp} and the fact that $\p^H = \np^H$, it follows that each $L \in \calC_{\text{compl.}}^B$ is also in $\p^B$ (by taking $A=H$ and $O = C \oplus E$).
    By Claim~\ref{claim:robust-equality}, it thus follows that $\p^B = \np^B \cap \conp^B$.

    By Claim~\ref{claim:coupf1-in-np}, $\coup_{\leq f_1(n)}^B \subseteq \np^B$. Equivalently, $\upleq{f_1(n)}^B \subseteq \conp^B$. By definition, $\upleq{f_1(n)}^B \subseteq \np^B$. Therefore $\p^B \subseteq \upleq{f_1(n)}^B \subseteq \np^B \cap \conp^B = \p^B$. 
    \end{proof}

    We have thus obtained an oracle with the correct properties.
    Moreover, if the functions $\floor{f_1(n)}$ and $\floor{f_2(n)}$ are both computable, then the oracle is recursive.

\end{document}